\documentclass[11pt]{article}

\usepackage[T1]{fontenc}
\usepackage{textcomp}
\usepackage{palatino}
\usepackage{mathpazo}
\usepackage{stmaryrd}
\usepackage{tabularx}

\usepackage{hyperref}
\hypersetup{pdfpagemode=UseNone}

\usepackage{amsfonts}
\usepackage{amssymb}
\usepackage{amsmath}
\usepackage{latexsym}
\usepackage{amsthm}
\usepackage{eepic}
\usepackage{sectsty}
\usepackage[usenames]{color}

\newtheorem{theorem}{Theorem}
\newtheorem{lemma}[theorem]{Lemma}

\theoremstyle{definition}
\newtheorem{definition}{Definition}

\newtheorem{problem}{Problem}

\newtheorem*{fact}{Fact}

\newtheorem{corollary}{Corollary}

\newcommand{\tinyspace}{\mspace{1mu}}
\newcommand{\microspace}{\mspace{0.5mu}}

\newcommand{\snorm}[1]{\lVert\tinyspace#1\tinyspace\rVert}

\newcommand{\ceil}[1]{\left\lceil #1 \right\rceil}

\newcommand{\tr}{\operatorname{Tr}}

\newcommand{\ip}[2]{\left\langle #1 , #2\right\rangle}

\def\({\left(}
\def\){\right)}
\def\I{\mathbb{1}}

\newcommand{\setft}[1]{\mathrm{#1}}
\newcommand{\lin}[1]{\setft{L}\left(#1\right)}
\newcommand{\density}[1]{\setft{D}\left(#1\right)}

\newcommand{\herm}[1]{\setft{Herm}\left(#1\right)}

\newcommand{\SepD}[1]{\setft{SepD}\left(#1\right)}

\newcommand{\ot}{\otimes}

\def\complex{\mathbb{C}}

\def\Rep{\text{Re}}

\def \lket {\left|}
\def \rket {\right\rangle}
\def \lbra {\left\langle}
\def \rbra {\right|}
\newcommand{\ket}[1]{\lket\microspace #1 \microspace\rket}
\newcommand{\bra}[1]{\lbra\microspace #1 \microspace\rbra}
\newcommand{\ketbra}[2]{\lket #1 \rangle \langle #2 \rbra}
\newcommand{\braket}[2]{\lbra #1 | #2 \rket}

\newenvironment{mylist}[1]{\begin{list}{}{
    \setlength{\leftmargin}{#1}
    \setlength{\rightmargin}{0mm}
    \setlength{\labelsep}{2mm}
    \setlength{\labelwidth}{8mm}
    \setlength{\itemsep}{0mm}}}
    {\end{list}}

\newcommand{\class}[1]{\textup{#1}}

\newcommand{\onenorm}[1]{\snorm{#1}_1}

\newcommand{\infnorm}[1]{\snorm{#1}_\infty}
\newcommand{\trnorm}[1]{\snorm{#1}_\text{tr}}
\newcommand{\fbnorm}[1]{\snorm{#1}_\text{F}}
\newcommand{\loccnorm}[1]{\snorm{#1}_\text{LOCC}}
\newcommand{\opnorm}[1]{\snorm{#1}_\text{op}}
\newcommand{\ball}[2]{\mathbf{B}(\complex^{#1},\snorm{\cdot}_{#2})}
\newcommand{\OptSep}[1]{\text{OptSep}(#1)}
\newcommand{\sspace}[1]{\mathrm{SP}(#1)}
\newcommand{\sraw}[1]{\text{Raw-(}#1\text{)}}
\newcommand{\citeequ}[1]{Equ.~(#1)}
\newcommand{\tensorrank}[1]{\text{rank}_\ot(#1)}
\newcommand{\btrank}[2]{\text{brank}_\ot(#1,#2)}

\def\X{\mathcal{X}}

\def\A{\mathcal{A}}
\def\B{\mathcal{B}}
\def\V{\mathcal{V}}

\def\H{\mathcal{H}}

\def\N{\mathcal{N}}
\def\R{\mathcal{R}}

\def\Q{\mathcal{Q}}

\def\nL{\ell}

\begin{document}

\title{\bf Epsilon-net method for optimizations over separable states}

\author{%
  Yaoyun Shi and Xiaodi Wu  \\
  \it \small Department of Electrical Engineering and Computer Science, University of Michigan, Ann Arbor, USA
}

\date{}

\maketitle

\begin{abstract}\noindent
We give algorithms for the optimization problem:
$\max_\rho \ip{Q}{\rho}$, where $Q$ is a Hermitian matrix, and the variable
$\rho$ is a bipartite {\em separable} quantum state. This problem lies at the heart of several problems in quantum
computation and information, such as the complexity of QMA(2).
While the problem is NP-hard, our algorithms are better than brute force
for several instances of interest. In particular, they give PSPACE upper bounds
on promise problems admitting a QMA(2)
protocol in which the verifier performs only logarithmic number of elementary gate
on both proofs, as well as the promise problem of deciding if a bipartite
local Hamiltonian has large or small ground energy.
For $Q\ge0$, our algorithm runs in time exponential in $\|Q\|_F$.
While the existence of such an algorithm was first proved recently by
Brand{\~a}o, Christandl and Yard [{\em Proceedings of the 43rd annual ACM Symposium on Theory of Computation} , 343--352, 2011], our algorithm is conceptually simpler.
\end{abstract}

\section{Introduction} \label{sec:intro}

Entanglement is an essential ingredient in many ingenious applications of
quantum information processing. Understanding and exploiting entanglement
remains a central theme in quantum information processing research~\cite{HHH+09}.
Denote by $\SepD{\A_1\ot \A_2}$ the set of separable (i.e, unentangled) density
operators over the space $\A_1\ot \A_2$. A fundamental question
 known as the \emph{weak membership} problem for separability
is to decide, given a classical description of a quantum state
$\rho$ over $\A_1 \ot \A_2$, whether this state $\rho$ is inside or
$\epsilon$ far away in trace distance from  $\SepD{\A_1 \ot \A_2}$. Unfortunately this very basic
problem turns out to be intractable. In 2003, Gurvits~\cite{Gur03}
proved the NP-hardness of the problem when $\epsilon$ is inverse
exponential in the dimension of $\A_1 \ot \A_2$. The dependence on
$\epsilon$ was later improved to inverse
polynomial~\cite{Ioa07,Gha10}.

In this paper we study a closely related problem to the weak
membership problem discussed above. More precisely, we consider the
linear optimization problem over separable states.
\begin{problem} \label{prob:opt_sep}
Given a Hermitian matrix $Q$ over $\A_1 \ot \A_2$ (of dimension
$d \times d$), compute the optimum value, denoted by $\OptSep{Q}$,
of the optimization problem
\[
      \max \ip{Q}{X}  \text{ subject to }  X \in \SepD{\A_1 \otimes \A_2}.
\]
\end{problem}
\noindent It is a standard fact in convex optimization~\cite{GLS93,
Ioa07} that the weak membership problem and the weak linear
optimization, a special case of Problem~\ref{prob:opt_sep}, over
certain convex set, such as $\SepD{\A_1 \ot \A_2}$, are equivalent
up to polynomial loss in precision and polynomial-time
overhead. Thus the hardness result on the weak membership problem
for separability passes directly to
Problem~\ref{prob:opt_sep}.

Besides the connection with the weak membership problem for
separability, Problem~\ref{prob:opt_sep} can also be understood from
many other aspects. Firstly, as the objective function is the
inner-product of a Hermitian matrix and a quantum state,  which represents the average value of some physical observable, the optimal value
of Problem~\ref{prob:opt_sep} inherently possesses certain physical
meaning.
Secondly, in the study of the tensor product space~\cite{DF92}, the value $\OptSep{Q}$ is precisely
the \emph{injective norm} of $Q$ in $\mathcal{L}(\A_1)\otimes\mathcal{L}(\A_2)$, where $\mathcal{L}(\A)$ denote the
Banach space of operators on $\A$ with the operator norm.
Finally, one may be equally motivated from the study in operations
research. The definition of Problem~\ref{prob:opt_sep} appeared in
an equivalent form in~\cite{LQNY} with the new name of
``Bi-Quadratic Optimization over Unit Spheres''. Subsequent
works~\cite{HLZ10,So11} demonstrate that Problem~\ref{prob:opt_sep}
is just a special case of a more general class of optimization
problems called homogenous polynomial optimization with quadratic
constraints, which is currently an active research topic in that
field.

Another motivation to study Problem~\ref{prob:opt_sep} is
the recent interest about the complexity class
\class{QMA(2)}. Originally the class \class{QMA} (also known as
\emph{quantum proofs}) was defined~\cite{KitaevW02} as the quantum
counterpart of the classical complexity class \class{NP}. While
the extension of \class{NP} to allow multiple provers trivially reduces to \class{NP} itself,
the power of $\class{QMA(2)}$, the extension for \class{QMA} with multiple {\em unentangled}
provers, remains far from being well understood.
The study of the multiple-prover model was initiated
in~\cite{KMY01,KMY03}, where \class{QMA(k)} denotes the complexity
class for the $k$-prover case. Much attention was
attracted to this model because of the discovery that NP
admits \emph{logarithmic}-size unentangled quantum
proofs~\cite{BT09}. This result was surprising because
single prover quantum logarithm-size proofs only characterize
\class{BQP}~\cite{MW05}. It seems adding one unentangled prover
increases the power of the model substantially. There are several
subsequent works on refining the initial protocol either with
improved completeness and soundness bounds~\cite{Bei10,ABD+09, CF11, GNN11} or
with less powerful verifiers~\cite{CD10}. Recently it was proved that
\class{QMA(2)}=\class{QMA(poly)} \cite{HM10} by using the so-called
\emph{product test} protocol that determines whether a multipartite
state is a product state when two copies of it are given.
There
is another line of research on the power of unentangled quantum
proofs with restricted verifiers. Two complexity classes
\class{BellQMA} and \class{LOCCQMA}, referring to the restricted
verifiers that perform only nonadaptive or adaptive local
measurements respectively, were defined in~\cite{ABD+09} and
studied in~\cite{Bra08,BCY11}. It has been shown~\cite{BCY11} that
\class{LOCCQMA(m)} is equal to \class{QMA} for constant $m$.

Despite much effort, no nontrivial upper bound of \class{QMA(2)}
is known. The best known upper bound \class{QMA(2)}$\subseteq$\class{NEXP}
follows trivially by nondeterministically guessing the two proofs. It
would be surprising if $\class{QMA(2)}=\class{NEXP}$.
Thus it is reasonable to seek a better upper bound like
\class{EXP} or even \class{PSPACE}. It is not hard to see that simulating
QMA(2) amounts to distinguishing between two promises of $\OptSep{Q}$,
although one has the freedom to choose the appropriate $Q$. Note
that Problem~\ref{prob:opt_sep} was also studied in~\cite{BCY11} for
the same purpose.

\vspace{3mm} \noindent \textbf{Hardness result.}
There are several approaches to prove the hardness of
Problem~\ref{prob:opt_sep}. The first is to make use of the
NP-hardness of the weak membership problem and the folk theorem in
convex optimization as mentioned above. However, one may directly
reduce the \textrm{CLIQUE} problem to
Problem~\ref{prob:opt_sep}~\cite{deK08,LQNY}. There is also a
stronger hardness result~\cite{HM10} on the exact running time
of algorithms solving Problem~\ref{prob:opt_sep} conditioned on the
Exponential Time Hypothesis (ETH)~\cite{IP01}. The hardness results
extend naturally to the approximation version of
Problem~\ref{prob:opt_sep}. It is known that $\OptSep{Q}$ remains to
be NP-hard to compute even if inverse polynomial additive error is
allowed. Nevertheless, it is wide open whether the hardness result
remains if one allows even larger additive error.

From the perspective of operations research, the hardness of Problem 1 is  a consequence of not being a convex optimization problem.
In this case although efficient methods, compared with brute-force, for finding a local optimum
usually exist, on the other hand finding the global one is fraught with difficulty.
This is because one needs to enumerate all possible local optima
before one can determine the global optimum in the worst case.

\vspace{3mm} \noindent \textbf{Our contributions.} In this paper we
provide efficient algorithms for Problem~\ref{prob:opt_sep}
in either time or space for several $Q$s of interest. As the hardness result implies that enumeration is likely to
be inevitable in the worst case, our idea is to enumerate via
epsilon-nets more "cleverly" with the help of certain structure of
$Q$.

When the total number of points to enumerate is not large, one can
represent and hence enumerate each point in polynomial space. If the
additional computation for each point can also be done in polynomial
space, one immediately gets a polynomial-space implementation for
the whole algorithm by composing those two components naturally. We
make use of the relation \class{NC(poly)}=\class{PSPACE}
\cite{Borodin77} to obtain space-efficient implementation for the
additional computation, which in our cases basically includes the
following two parts. The first part helps to make sure the
enumeration procedure works correctly. This is because these
epsilon-nets of interest in our algorithm are not standard,
additional effort is necessary to generate them. This part turns
into a simple application of the so-called \emph{multiplicative
matrix weight update} (MMW)
method~\cite{AroraHK05,WarmuthK06,Kale07} to computing a min-max
form, which is known to admit efficient parallel algorithms under
certain conditions. The second part contains the real computation
which ,in our case, only consists of fundamental matrix operations.
It is well known those operations usually admit efficient parallel
algorithms~\cite{vzGathen93}. As a result, both parts of the
additional computation admit efficient parallel algorithms, and
therefore, the additional computation can be implemented in
polynomial space in our case.

We summarize below the main results obtained by
applying the above ideas.

\vspace{2mm} \noindent \textbf{1}. The first property exploited is
the so-called \emph{decomposability} of $Q$ which refers to whether
$Q$ can be decomposed in the form $Q=\sum_{i=1}^M Q^1_i \ot Q^2_i$
with small $M$. Note this concept is closely related to a more
commonly studied concept, \emph{tensor rank}. Intuitively, if one
substitutes this $Q$'s decomposition into $\ip{Q}{ \rho_1 \ot
\rho_2}$ and treat
$\ip{Q^1_1}{\rho_1},\cdots,\ip{Q^1_M}{\rho_1},\ip{Q^2_1}{\rho_2},\cdots,
\ip{Q^2_M}{\rho_2}$ as variables, the optimization problem becomes
quadratic and $M$ corresponds to the number of second-order
terms in the objective function. If we plug the values of $\ip{Q^1_1}{\rho_1},\cdots,\ip{Q^1_M}{\rho_1}$ into the objective function, then the optimization problem reduces to be a semidefinite program, and thus can be efficiently solved. Hence by enumerating all possible values of $\ip{Q^1_1}{\rho_1},\cdots,\ip{Q^1_M}{\rho_1}$ one can efficiently solve the original problem when $M$ is small.
Since this approach naturally extends to the
$k$-partite case for $k\geq 2$, we obtain the following general result.

\begin{theorem}[Informal. See Section~\ref{sec:main}]
Given any Hermitian $Q$ and its decomposition $Q=\sum_{i=1}^M Q^1_i
\ot \cdots \ot Q^k_i$ as input, the quantity $\OptSep{Q}$ can be
approximated with additive error $\delta$ in quasi-polynomial
time\footnote{Quasi-polynomial time is upper bounded by $2^{O((\log
n)^c)}$ for some fixed c, where $n$ is the input size.} in $d$ and
$1/\delta$ if $kM$ is bounded by poly-logarithms of $d$.
\end{theorem}

By exploiting the space-efficient algorithm design strategy above,
this algorithm can also be made space-efficient. To facilitate the
later applications to complexity classes, we choose the input size
to be some $n$ such that $d=\exp(\text{poly}(n))$.

\begin{corollary}[Informal. See Section~\ref{sec:main}]
If $kM/\delta \in O(\text{poly}(n))$, the quantity $\OptSep{Q}$ can
be approximated with additive error $\delta$ in \class{PSPACE}.
\end{corollary}

As a direct application, we prove the following variant of
\class{QMA(2)} belongs to \class{PSPACE} where
$\class{QMA(2)}[\text{poly}(n), O(\log(n))]$ refers to the model
where the verifier only performs $O(\log(n))$ elementary gates that act on
both proofs at the same time and a polynomial number of other elementary
gates. Note \class{QMA(2)}[poly(n),poly(n)]=\class{QMA(2)} in
our notation.
\begin{corollary}
$\class{QMA(2)}[\text{poly}(n), O(\log(n))]$ $\subseteq
\class{PSPACE}$.
\end{corollary}
This result establishes the first \class{PSPACE} upper bound for a
variant of \class{QMA(2)} where the verifier is allowed to generate
some quantum entanglement between two proofs. In contrast, previous
results are all about variants with nonadaptive or adaptive local
measurements, such as  \class{BellQMA(2)}~\cite{ABD+09,Bra08,CD10}
or \class{LOCCQMA(2)}~\cite{ABD+09, BCY11}.

We also study Problem~\ref{prob:opt_sep} when $Q$ is a local
Hamiltonian over $k$ parties. Recall that a promise version of this
problem in the one party case, namely the \emph{local-Hamiltonian
problem}, is \class{QMA}-complete problem~\cite{KitaevW02}. Our
definition extends the original local Hamiltonian problem to its
$k$-partite version. However, as will be clear in the main section,
the $k$-partite local Hamiltonian problem is no longer necessarily
\class{QMA(k)}-complete. On the other side, our enumeration
algorithm based on the decomposability of $Q$ works extremely well
in this case. As a result, we obtain the following corollary.

\begin{corollary}[Informal. See Section~\ref{sec:local_hamiltonian}]
Given some local Hamiltonian $Q$ over $k$ parties as input,
$\OptSep{Q}$ can be approximated with additive error $\delta$ in
quasi-polynomial time in $d,1/\delta$; the $k$-partite local
Hamiltonian problem belongs to \class{PSPACE}.
\end{corollary}

Very recently, an independent result~\cite{CS11} of us shows that
the 2-partite local Hamiltonian problem defined above lies in
\class{QMA}, and henceforth in \class{PSPACE}, which complements our
algorithmic result.

\vspace{2mm} \noindent \textbf{2.} The second structure made use of
is the eigenspace of $Q$ of large eigenvalues. As a result, we
establish an algorithm solving Problem~\ref{prob:opt_sep} with
running time exponential in $\fbnorm{Q}$.

\begin{theorem}[Informal. See Section~\ref{sec:alg_Q2}]
For positive semidefinite $Q$,  the quantity $\OptSep{Q}$ can be
approximated with additive error $\delta$ in time
$\exp(O(\log(d)+\delta^{-2} \fbnorm{Q}^2 \ln (\fbnorm{Q}/\delta)))$.
\end{theorem}
A similar running time $\exp(O(\log^2(d)\delta^{-2}\fbnorm{Q}^2 ))$
was obtained in~\cite{BCY11} using some known results in quantum
information theory.(i.e., the semidefinite programming for finding
symmetric extension~\cite{DPS04} and an improved quantum de
Finetti-type bound.) In contrast, our algorithm only uses
fundamental operations of matrices and epsilon-nets. To approximate
with precision $\delta$, it suffices to consider the eigenspace of
$Q$ of eigenvalues greater than $\delta$ whose dimension  is bounded
by $\fbnorm{Q}^2/\delta^2$. Nevertheless, naively enumerating
density operators over that subspace does not work since one cannot
detect the separability of those density operators. We circumvent
this difficulty by making nontrivial use of the Schmidt
decomposition of bipartite pure states.

We note, however, that other results in~\cite{BCY11} do not follow
from our algorithm, and our method cannot be seen as a replacement
of the kernel technique therein. Furthermore, our method does not
extend to the $k$-partite case, as there is no Schmidt decomposition
in that case.

\vspace{3mm} \noindent \textbf{Open problems.} The main open problem
is whether Problem~\ref{prob:opt_sep} admits an efficient algorithm
in either time or space, when larger additive error is allowed. It is
also interesting to see whether, for those $Q$s that come from the
simulation of the complexity class \class{QMA(2)}, the quantity
$\OptSep{Q}$ can be efficiently computed.

\vspace{3mm} \noindent \textbf{Organizations:} The rest part of this
paper is organized as follows. The necessary background knowledge on
the particular epsilon-nets in use is introduced in
Section~\ref{sec:epsilon_net}. The main algorithm based on the
decomposability of $Q$ is illustrated in Section~\ref{sec:main}. Two applications of such an algorithm is discussed immediately after;
 the simulation of variants of \class{QMA(2)} is discussed in
Section~\ref{sec:QMA_variants} and the local Hamiltonian case is
discussed in Section~\ref{sec:local_hamiltonian}. Finally, the
demonstration of an algorithm with running time exponential in
$\fbnorm{Q}$ for Problem~\ref{prob:opt_sep} can be found in
Section~\ref{sec:alg_Q2}.

\vspace{3mm} \noindent \textbf{Notations:} We assume familiarity
with standard concepts from quantum
information~\cite{NielsenC00,KitaevW02,Watrous08-lec}. Particularly,
our notations follow from~\cite{Watrous08-lec}. Precisely, we use
$\A,\B$ to denote complex Euclidean spaces and
$\lin{\A},\herm{\A},\density{\A}$ to stand for the linear operators,
Hermitian operators and density operators over $\A$
respectively. We denote the trace norm of operator $Q$ by
$\trnorm{Q}$, i.e. $\trnorm{Q}=\tr{(Q^* Q)^{1/2}}$ where $Q^*$
stands for the conjugate transpose of $Q$. The Frobenius norm is
denoted by $\fbnorm{Q}$ and the operator norm is denoted by
$\opnorm{Q}$. The $\nL_1$ norm of vector $x \in \complex^n$ is
denoted by $\onenorm{x}=\sum_{i=1}^n |x_i|$ and its $\nL_\infty$
norm is denoted by $\infnorm{x}=\max_{i=1,\cdots,n} |x_i|$. We use
$\snorm{x}$ to denote the Euclidean norm.  The unit ball of
$\complex^n$ under certain norm $\snorm{\cdot}$ is denoted by
$\ball{n}{}$.

\section{Epsilon Net} \label{sec:epsilon_net}
The epsilon-net (or $\epsilon$-net) is an important concept  in several
mathematical topics.
For our purpose, we adopt the following definition of $\epsilon$-net.

\begin{definition}[$\epsilon$-net]
\label{def:e_nets} Let $(X,d)$ \footnote{ We will abuse the notation
later where the metric $d$ is replaced by the norm from which the
metric is induced.} be any metric space and let $\epsilon>0$. A
subset $\N_\epsilon$ is called an $\epsilon$-net of $X$ if
for each $x\in X$, there exists $y \in \N_\epsilon$
with $d(x, y)\leq \epsilon$.
\end{definition}


Now we turn to the particular $\epsilon$-net considered in this
paper. Let $\H$ be any Hilbert space of dimension $d$ and
$\mathcal{Q}=\mathcal{Q}(M,w)=(Q_1,Q_2,\cdots Q_M)$ be a sequence
of operators on $\H$ with $\opnorm{Q_i}\le w$, for all $i$.
Define the $\mathcal{Q}$-space, denoted by $\sspace{\Q}$,  as
\[
 \sspace{\Q} = \{ ( \ip{Q_1}{\rho}, \ip{Q_2}{\rho},
 \cdots, \ip{Q_M}{\rho}): \rho \in \density{\H} \} \subseteq \mathbb{C}^M.
\]
The set is convex and compact, and a (possibly proper) subset of
$\sraw{M,w}=\{(q_1,q_2,\cdots, q_M): \forall i, q_i \in \mathbb{C},
\snorm{q_i}\leq w\}$.

In the following, we construct an $\epsilon$-net of the metric space $(\sspace{\Q}, \nL_1)$. Our
method will first generate an $\epsilon$-net of $(\sraw{M,w},
\nL_1)$ via a standard procedure and then
select those points that are also close to $\Q$-space.
We will present and analyze the efficiency of the selection process first
and come back to the construction of the $\epsilon$-net afterwards.

\subsection*{Selection process}
The selection process determines if some point $\vec{p}$ in
$\sraw{M,w}$ is close to  $\sspace{\Q}$. 
Denote by $\text{dis}(\vec{p})$ the distance of
$\vec{p}\in\mathbb{C}^M$ to $\sspace{\Q}$, i.e.,
\[
  \text{dis}(\vec{p}) = \min_{ \vec{q} \in \sspace{\Q}}
  \snorm{\vec{p}-\vec{q}}_1.
\]
We show in this section how to compute $\text{dis}(\vec{p})$
efficiently in space. That the problem admits a polynomial time
algorithm follows from the fact that it can be cast as a
semidefinite programming problem. However, to the authors'
knowledge, only a few restricted classes of SDPs also admit
space-efficient algorithms and none of them applies to our case.
Thus we need to develop our own space-efficient algorithm for this
problem.

By making use of the definition of $\sspace{\Q}$ and the duality
of the $\nL_1$ norm, one can find the following equivalent
definition of the distance.
\begin{equation*} 
\text{dis}(\vec{p})= \min_{\rho \in \density{\H}} \max_{\vec{z} \in
\ball{M}{\infty}} \Rep \ip{\vec{p}-\vec{q}(\rho)}{\vec{z}},
\end{equation*}
where
\begin{equation} \label{eqn:def_vec_q}
 \vec{q}(\rho)=(\ip{Q_1}{ \rho}, \ip{Q_2}{ \rho}, \cdots, \ip{Q_M}{\rho})
 \in \complex^M.
\end{equation}
By rephrasing $\text{dis}(\vec{p})$ in the above form, one shows
the quantity $\text{dis}(\vec{p})$ is actually an equilibrium value. This follows
from the well-known extensions of von' Neumann's Min-Max
Theorem~\cite{vN28,Fan53}. One can easily verify that the density
operator set $\density{\H}$ and the unit ball of $\complex^M$ under
$\nL_\infty$ norm are convex and compact sets. Moreover, the
objective function is a bilinear form over the two sets. The Min-Max
theorem implies
\begin{equation} \label{eqn:closeness_equilibrium}
\min_{\rho \in \density{\H}} \max_{\vec{z} \in \ball{M}{\infty} }
\Rep \ip{\vec{p}-\vec{q}(\rho)}{\vec{z}} = \max_{\vec{z} \in
\ball{M}{\infty}} \min_{\rho \in \density{\H}} \Rep
\ip{\vec{p}-\vec{q}(\rho)}{\vec{z}}.
\end{equation}

Fortunately, there is an efficient algorithm in either time or space
(in terms of $d,M,w,1/\epsilon$) to approximate
$\text{dis}(\vec{p})$ with additive error $\epsilon$. The main tool
used here is the so-called matrix multiplicative weight update
method~\cite{AroraHK05,Kale07,WarmuthK06}. Similar min-max forms
also appeared before in a series of work on quantum
complexity~\cite{JainW09,Wu10a,Wu10b,GutoskiW10}. The algorithm
presented here is another simple application of this powerful
method. For the sake of completeness, we provide the proof of the
following lemma in Appendix~\ref{app:proof_lm_close}.

\begin{lemma} \label{lm:closeness_check}
Given any point $\vec{p} \in \sraw{M,w}$ and $\epsilon>0$, there is
an algorithm (depicted in Appendix~\ref{app:proof_lm_close}) that
approximates $\text{dis}(\vec{p})$ with additive error $\epsilon$.
Namely, the return value $\tilde{d}$ of this algorithm satisfies
\[
   \tilde{d}-\epsilon \leq \text{dis}(\vec{p}) \leq
   \tilde{d}+\epsilon.
\]
Moreover, the algorithm runs in $\text{poly}(d,M,w,1/\epsilon)$
time. Furthermore, if $d$ is considered as the input size and
$M,w,1/\epsilon \in O( \text{poly-log}(d))$, this algorithm is also
efficient in parallel, namely, it is inside \class{NC}.
\end{lemma}

\subsection*{Construction of the $\epsilon$-net}
We are now ready to show the construction of the $\epsilon$-net of
$\sspace{\Q}$. As mentioned before, this construction contains
two steps below. Given any $\Q(M,w)$ and $\epsilon>0$,
\begin{itemize}
 \item Construct the $\epsilon$-net of the set $\sraw{M,w}$ with the metric induced from the $\nL_1$ norm. Denote such an $\epsilon$-net by $\R_\epsilon$.
 \item For each point $\vec{p} \in\R_\epsilon$, determine
 $\text{dis}(\vec{p})$ and select it to $\N_\epsilon$ if
 dis$(\vec{p})\leq \epsilon$. We claim $\N_\epsilon$ is the
 $\epsilon$-net of $(\sspace{\Q}, \nL_1)$.
\end{itemize}
The construction for the first step is rather routine. Creating an
$\epsilon'$-net $T_\epsilon'$ over a bounded complex region $\{z\in
\mathbb{C}: \snorm{z}\leq w\}$ is simple: we can place a 2D grid
over the complex plane to cover the disk $\snorm{z}\leq w$.  Simple
argument shows $|T_\epsilon'| \in O(\frac{w^2}{\epsilon'^2})$.  Then
$\R_\epsilon$ can be obtained by the cross-product
$\underbrace{T_\epsilon'\times \cdots \times T_\epsilon'}_{M \text{
times}}$. To ensure the closeness in the $\nL_1$ norm, we will choose
$\epsilon'=\epsilon/M$.

\begin{theorem} \label{thm:epsilon_net_main}
The $\N_\epsilon$ constructed above is indeed an $\epsilon$-net of
$(\sspace{\Q}, \nL_1)$ with cardinality at most
$O((\frac{w^2M^2}{\epsilon^2})^M )$. Furthermore, for any point
$\vec{n} \in \N_\epsilon$, we have $\text{dis}(\vec{n})\leq
\epsilon$.
\end{theorem}

\begin{proof}
First we show $\R_\epsilon$ is indeed an $\epsilon$-net of
$(\sraw{M,w}, \nL_1)$. To that end, consider any point $\vec{p} \in
\sraw{M,w}$. From the construction of $\R_\epsilon$, there is some
point $\vec{q} \in \R_\epsilon$ such that
$\snorm{\vec{p}-\vec{q}}_\infty \leq \epsilon'$. Then we have
$\snorm{\vec{p}-\vec{q}}_1 \leq M \snorm{\vec{p}-\vec{q}}_\infty
\leq M \epsilon' \leq \epsilon$. Since $\N_\epsilon \subseteq
\R_\epsilon$, one has $|\N_\epsilon| \leq |\R_\epsilon| \in
O((\frac{w^2M^2}{\epsilon^2})^M)$.

In order to show $\N_\epsilon$ is the required $\epsilon$-net,
consider any point $\vec{p} \in \sspace{\Q}$. Since $\sspace{\Q}
\subseteq \sraw{M,w}$, there exists a point $\vec{p'} \in
\R_\epsilon$ such that $\snorm{\vec{p}-\vec{p'}}_1 \leq \epsilon$.
Hence we have $\text{dis}(\vec{p'})\leq \epsilon$ and the point
$\vec{p'}$ will be selected, namely $\vec{p'} \in \N_\epsilon$.
Finally, it is a simple consequence of the selection process that
every point $\vec{n} \in \N_\epsilon$ has $\text{dis}(\vec{n})\leq
\epsilon$ .
\end{proof}

\noindent \textbf{Remarks.} If one choose $\Q$ to be $
 \Q(d^2, 1)=\{\ketbra{i}{j}: i,j=1,\cdots, d\}$, one can generate the $\epsilon$-net of the density operator set
with the $\nL_1$ norm in the method described above. It is akin to generating
an $\epsilon$-net for every entry of the density operator. At the other
extreme, one can also efficiently generate the $\epsilon$-net of
a small size $\sspace{\Q}$ even when the space dimension $d$ is
relatively large.

\section{The Main Algorithm} \label{sec:main}

In this section, we prove the main theorem.
Without loss of generality, we assume $\A_1,\A_2$ are
identical, and of dimension $d$ in Problem~\ref{prob:opt_sep}.
Moreover, our algorithm will deal with the set of product states
rather than separable states. Namely, we consider the following optimization problem.
\begin{align}
    \text{max:}\quad & \ip{Q}{\rho} \label{eqn:def_opt_product_problem}\\
    \text{subject to:}\quad & \rho=\rho_1 \otimes \rho_2,  \rho_1 \in \density{\A_1}, \rho_2 \in \density{\A_2}. \nonumber
\end{align}
It is easy to see these two optimization problems are equivalent
since product states are extreme points of the set of separable
states. Our algorithm works for both maximization and minimization
of the objective function. In fact, both results can be obtained at
the same time. Since our algorithm naturally extends to multipartite
cases, we will demonstrate the algorithm for the $k$-partite version
first, and then obtain the solution for Problem~\ref{prob:opt_sep} as
a special case when $k=2$.

\begin{problem}[k-partite version] \label{prob:opt_k_party}
Given any Hermitian matrix $Q$ over  $\A_1 \ot \cdots \ot \A_k$
($k\geq2$), compute the optimum value $\OptSep{Q}$ of the following
optimization problem to precision $\delta$.
\begin{align}
      \text{max:}\quad & \ip{Q}{\rho} \\
      \text{subject to:}\quad & \rho=\rho_1 \ot \cdots \ot \rho_k, \forall i, \rho_i \in \density{\A_i}. \nonumber
\end{align}
\end{problem}

\noindent Before describing the algorithm, we need some terminology
about the \emph{decomposability} of a multi-partite operator. Any
Hermitian operator $Q$ over $\A_1 \ot \A_2 \ot \cdots \ot
\A_k$ is called \emph{$M$-decomposable} if there exists
$(Q^t_1,Q^t_2, \cdots, Q^t_M) \in \lin{\A_t}^M$ for t=1,2,...,
k such that
\[
  Q=\sum_{i=1}^M Q^1_i \ot Q^2_i \ot \cdots \ot Q^{k-1}_i \ot Q^k_i.
\]
To facilitate the use of $\epsilon$-net, we adopt a slight variation
of the decomposability above. Let $\vec{w} \in \mathbb{R}^k_{+}$
denote the widths of operators over each $\A_i$. Any $Q$ is called
\emph{$(M, \vec{w})$-decomposable} if $Q$ is $M$-decomposable and
the widths of those operators in the decomposition are bounded in
the sense that $\max_i \opnorm{Q^t_i}\leq w_t$ for t=1,2,..., k. It
is noteworthy to mention that the decomposability defined above is
related to the concept tensor rank~\footnote{Our definition should
be more accurately related to the Kronecker-Product rank defined in
\cite{RvL11}, a special case of the more general concept tensor
rank. } defined in tensor product spaces. Precisely for any
hermitian operator $Q$ over $\A_1 \otimes \A_2 \ot \cdots \ot \A_k$,
its tensor rank \emph{$\tensorrank{Q}$} is defined to be $\min\{M| Q
\text{ is }M\text{-decomposable}\}$. Its bounded tensor rank
\emph{$\btrank{Q}{\vec{w}}$} is defined to be $\min\{M| Q \text{ is
}(M,\vec{w})\text{-decomposable}\}$.


By definition, we have $\tensorrank{Q}$ (resp.
$\btrank{Q}{\vec{w}}$) is the minimum $M$ that $Q$ can be $M$ (resp.
($M, \vec{w}$))-decomposable. However, given the representation $Q$
as input, it is hard in general to compute $\tensorrank{Q}$,
$\btrank{Q}{\vec{w}}$, or its corresponding decomposition. Therefore
it is hard to make use of the optimal decomposition when $Q$ is the
only input. Instead, for any $(M,\vec{w})$-decomposable $Q$ we
assume its corresponding decomposition is also a part of the input
to our algorithm.

\begin{figure}[t]
\noindent\hrulefill
\begin{mylist}{8mm}
\item[1.]
Let $\Q_t(M,w_t)=(Q^t_1,Q^t_2,\cdots, Q^t_M)$ for t=1,..., k-1. Let
  $W=\Pi_{i=1}^k w_i$. Generate the $\epsilon_t$-net (by Theorem~\ref{thm:epsilon_net_main}) of
  $(\sspace{\Q_t}, \nL_1)$ for each t=1,..., k-1 with $\epsilon_t=w_t \delta/(k-1)W$ and
  denote such a set by $\N^t_{\epsilon_t}$. Also let $\text{OPT}$ store
the optimum value of the maximization problem.
\item[2.]
For each point $\vec{q}=(\vec{q}^1, \vec{q}^2,\cdots \vec{q}^{k-1})
  \in \N^1_{\epsilon_1} \times \N^2_{\epsilon_2} \times \cdots
  \times
  \N^{k-1}_{\epsilon_{k-1}}$, let $Q^{k}$ be
  \[
     Q^{k}=\sum_{i=1}^M q^1_i q^2_i \cdots q^{k-1}_i Q^k_i,
  \]
  and calculate $\tilde{Q}^k=\frac{1}{2} (Q^k+Q^{k*})$. Then compute the
maximum eigenvalue of $\tilde{Q}^k$, denoted by
$\lambda_{\max}(\vec{q})$. Update OPT as follows: $ \text{OPT}=
\max\{\text{OPT}, \lambda_{\max}(\vec{q})\}$.

\item[3.]
Return $\text{OPT}$.
\end{mylist}
\noindent\hrulefill \caption{The main algorithm with precision
$\delta$. } \label{fig:main_alg}
\end{figure}

\begin{theorem} \label{thm:main_alg}
Let $Q$ be some $(M,\vec{w})$-decomposable Hermitian over
$\A_1\otimes \A_2 \ot \cdots \ot \A_k$ (each $\A_i$ is of dimension
$d$) and $\delta>0$. Also let $(Q^t_1,Q^t_2,\cdots, Q^t_M)
,t=1,2,\cdots, k$ be the operators in the corresponding
decomposition of $Q$. The algorithm shown in Fig.~\ref{fig:main_alg}
approximates the optimum value $\OptSep{Q}$ of
Problem~\ref{prob:opt_k_party} with additive error $\delta$.
Furthermore, the whole algorithm runs in
$O((\frac{(k-1)^2W^2M^2}{\delta^2})^{(k-1)M} ) \times\text{poly}(d,
M, k, W, 1/\delta)$ time.
\end{theorem}

\begin{proof}
Let's first prove the correctness of the algorithm. By choosing
$\Q_t(M,w_t)=(Q^t_1,Q^t_2,\cdots, Q^t_M)$ for t=1,...,k-1, the
algorithm first generates the $\epsilon$-net $\N_{\epsilon_t}$ of
each $(\sspace{\Q_t}, \nL_1)$ , whose correctness is guaranteed by
Theorem~\ref{thm:epsilon_net_main}. By substituting the identity
$Q=\sum_{i=1}^M Q^1_i \otimes Q^2_i \ot \cdots \ot Q^{k-1}_i$, the
optimization problem becomes
\begin{align*}
    \text{max:}\quad & \ip{\sum_{i=1}^M p^1_ip^2_i\cdots p^{k-1}_i Q^k_i}{\rho_k} \\
    \text{subject to:}\quad & \forall t \in \{1,\cdots, k-1\}, \vec{p_t} \in \sspace{\Q_t(M,w_t)}, \text{ and } \rho_k \in
    \density{\A_k}.
\end{align*}
Thus, solving the optimization problem amounts to first enumerating
$\vec{p}_t \in \sspace{\Q_t(M,w_1)}$ for each $t$, and then solving
the optimization problem over $\density{\A_k}$.

Consider any point $\vec{p}=(\vec{p}^1,\vec{p}^2,\cdots,
\vec{p}^{k-1}) \in \sspace{\Q_1}\times \cdots \times
\sspace{\Q_{k-1}}$. Due to Theorem~\ref{thm:epsilon_net_main}, there
is at least one point $\vec{q}=(\vec{q}^1, \vec{q}^2,\cdots
\vec{q}^{k-1})
  \in \N^1_{\epsilon_1} \times \N^2_{\epsilon_2} \times
  \cdots \times \N^{k-1}_{\epsilon_{k-1}}$ such that $\onenorm{\vec{q}^t-\vec{p}^t}\leq \epsilon_t$ for t=1,..,k-1.
The choice of $\tilde{Q}^k$ is to symmetrize $Q^k$ where the
latter is not guaranteed to be Hermitian because $\vec{q}$ only
comes from an $\epsilon$-net. With $\tilde{Q}^k$ being Hermitian, it
is clear that $\lambda_{\max}(\vec{q})=\max_{\rho_k \in
\density{\A_k}} \ip{\tilde{Q}^k}{\rho_k}$. Now let's analyze how
much error will be induced in this process.

Let $P^k(\vec{p})=\sum_{i=1}^M p^1_ip^2_i\cdots p^{k-1}_i Q^k_i$ and
$\tilde{P}^k=\frac{1}{2}(P^k+P^{k*})$. It is not hard to see that
$P^k=\tilde{P}^k$. The error bound is achieved by applying a chain
of triangle inequalities as follows. Firstly, one has
\[
\opnorm{\tilde{P}^k-\tilde{Q}^k}  =  \opnorm{\frac{1}{2} (
  P^k-Q^k) + \frac{1}{2} (P^{k*}-Q^{k*})} \leq \frac{1}{2} (\opnorm{P^k-Q^k}+ \opnorm{P^{k*}-Q^{k*}}  ) \\
    =  \opnorm{P^k-Q^k}.
\]
Then we substitute the expressions for $P^k, Q^k$ and apply the
standard hybrid argument.
\begin{eqnarray*}
 \opnorm{P^k-Q^k}& = & \opnorm{\sum_{i=1}^M (p^1_ip^2_i\cdots p^{k-1}_i-q^1_iq^2_i\cdots q^{k-1}_i) Q^k_i}
   \\
   & = & \opnorm{\sum_{i=1}^M \sum_{t=1}^{k-1} (q^1_i \cdots q^{t-1}_i p^t_i p^{t+1}_i \cdots p^{k-1}_i-q^1_i \cdots q^{t-1}_i q^t_i p^{t+1}_i \cdots p^{k-1}_i)
   Q^k_i},
\end{eqnarray*}
which is immediately upper bounded by the sum of the following
terms,
\[
  \sum_{i=1}^M |p^1_i-q^1_i| |p^2_i \cdots p^{k-1}_i|
\opnorm{Q^k_i},  \sum_{i=1}^M |q^1_i| |p^2_i-q^2_i| |p^3_i \cdots
p^{k-1}_i|
   \opnorm{Q^k_i}
     , \cdots , \sum_{i=1}^M |q^1_i \cdots q^{k-2}_i| |p^{k-1}_i-q^{k-1}_i|
   \opnorm{Q^k_i}.
\]
As the $t^{\text{th}}$ term above can be upper bounded by
$\epsilon_t W/w_t$ for each t=1,...,k-1, we have,
\[
\opnorm{\tilde{P}^k-\tilde{Q}^k} \leq \epsilon_1 W/w_1+ \epsilon_2
W/w_2 +\cdots  + \epsilon_{k-1} W/w_{k-1} =
\underbrace{\frac{\delta}{k-1} +\cdots+
   \frac{\delta}{k-1}}_{\text{k-1 terms}} =\delta.
\]
Hence the optimum value for any fixed $\vec{p}$ won't differ too
much from the one for its approximation $\vec{q}$ in the
$\epsilon$-net. This is because
\[
\max_{\rho_k \in \density{\A_k}} \ip {\tilde{P}^k}{\rho_k}  =
\max_{\rho_k \in \density{\A_k}} \ip
  {\tilde{Q}^k}{\rho_k} + \ip{\tilde{P}^k-\tilde{Q}^k}{\rho_k}.
\]
By H\"{o}lder Inequalities we have
$|\ip{\tilde{P}^k-\tilde{Q}^k}{\rho_k}| \leq
\opnorm{\tilde{P}^k-\tilde{Q}^k}\trnorm{\rho_k} \leq \delta$ and
thus,
\[
   \lambda_{\max}(\vec{q}) -\delta \leq \max_{\rho_k \in \density{\A_k}} \ip {\tilde{P}^k(\vec{p})}{\rho_k} \leq
  \lambda_{\max}(\vec{q}) +\delta.
\]
We now optimize $\vec{p}$ over $\sspace{\Q_1}\times \cdots \times
\sspace{\Q_{k-1}}$ and the corresponding $\vec{q}$ will run over the
$\epsilon$-net $\N^1_{\epsilon_1} \times \N^2_{\epsilon_2} \times \cdots \times
  \N^{k-1}_{\epsilon_{k-1}}$. As every point
$\vec{q} \in \N^1_{\epsilon_1} \times \N^2_{\epsilon_2} \times \cdots \times
  \N^{k-1}_{\epsilon_{k-1}}$ is also close to $\sspace{\Q_1}\times \cdots
\times \sspace{\Q_{k-1}}$ in the sense that
$\text{dis}(\vec{q}^t)\leq \epsilon_t$ for each t, we have
\[
  \max_{\vec{q} \in \N^1_{\epsilon_1} \times \N^2_{\epsilon_2} \times \cdots \times
  \N^{k-1}_{\epsilon_{k-1}}} \lambda_{\max}(\vec{q}) -\delta \leq \max_{\vec{p} \in \sspace{\Q_1}\times \cdots
\times \sspace{\Q_{k-1}}} \max_{\rho_k \in \density{\A_k}} \ip
{\tilde{P}^k(\vec{p})}{\rho_k} \leq
  \max_{\vec{q} \in \N^1_{\epsilon_1} \times \N^2_{\epsilon_2} \times \cdots \times
  \N^{k-1}_{\epsilon_{k-1}}} \lambda_{\max}(\vec{q}) +\delta.
\]
Finally, it is not hard to see that $\text{OPT}=\max_{\vec{q} \in
\N^1_{\epsilon_1} \times \N^2_{\epsilon_2} \times \cdots \times
  \N^{k-1}_{\epsilon_{k-1}}} \lambda_{\max}(\vec{q})$ and therefore
\[
 \text{OPT}-\delta \leq \OptSep{Q} \leq
 \text{OPT}+\delta.
\]
Now let us analyze the efficiency of this algorithm. The total
number of points in the $\epsilon$-net $\N^1_{\epsilon_1} \times
\N^2_{\epsilon_2} \times  \cdots \times \N^{k-1}_{\epsilon_{k-1}}$
is upper bounded by $O((\frac{(k-1)^2W^2M^2}{\delta^2})^{(k-1)M} )$
by Theorem~\ref{thm:epsilon_net_main}. For each point $\vec{q}$, the
generation of such a point will cost time polynomial in $d, M, W,
1/\delta$ (this part is done through the calculation of
$\text{dis}(\vec{q})$. See Lemma~\ref{lm:closeness_check}. ). After
the generation process, one needs to calculate $\tilde{Q}^k$ and its
maximum eigenvalue for each point, which can be done in time
polynomial in $d,k,M$.  Thus, the total running time is bounded by
 $O((\frac{(k-1)^2W^2M^2}{\delta^2})^{(k-1)M} )
\times\text{poly}(d, M, k, W, 1/\delta)$.
\end{proof}

\noindent \textbf{Remarks.} There are a few remarks to make about
Theorem~\ref{thm:main_alg}. First, it is straightforward to extend
the concept of decomposability  to its approximate version. For
instance, any Hermitian $Q$ is called \emph{$\epsilon$-approximate}
($M,\vec{w}$)-decomposable if there exists some
($M,\vec{w}$)-decomposable $\tilde{Q}$, such that
$\snorm{Q-\tilde{Q}} \leq \epsilon$, where the norm could be either
the operator norm or the \emph{injective} tensor norm. It is easy to
verify that the same algorithm solves $\OptSep{Q}$ approximately.

Second, all operations in the algorithm described in
Fig.~\ref{fig:main_alg} can be implemented efficiently in parallel
in some situation. This is because fundamental operations of
matrices can be done in \class{NC} and the calculation of
$\text{dis}(\vec{p})$ can be done in \class{NC} (See
Lemma~\ref{lm:closeness_check}) when $M, W, k, 1/\delta$ are in nice
forms of $d$. Thus, we can apply the observation stated in the
introduction and prove the algorithm in Fig.~\ref{fig:main_alg} can
also be made space-efficient. To facilitate the later use of this
result, we will change the input size as follows.

\begin{corollary} \label{cor:nc_main}
Let $n$ be the input size such that $d=\exp(\text{poly(n)})$, if
$W/\delta \in O(\text{poly}(n)), kM \in O(\text{poly(n)})$, then
$\OptSep{Q}$ can be approximated with additive error $\delta$ in
\class{PSPACE}.
\end{corollary}

\begin{proof}
Here we present an argument that composes space-efficient algorithms
directly.
Given $Q$ and its decomposition as input, consider the following algorithm
\begin{enumerate}
 \item Enumerate each point $\vec{p}=(\vec{p}_1, \vec{p}_2, \cdots, \vec{p}_{k-1})$ from the raw set $\R^1_{\epsilon_1}
 \times \cdots \times \R^{k-1}_{\epsilon_{k-1}}$.
 \item Compute $\text{dis}(\vec{p}_t)$ for each t=1,...,k-1. If
 $\vec{p}$ is a valid point in the epsilon-net, then we continue
 with the rest part in Step 2 of the algorithm in
 Fig.~\ref{fig:main_alg}.
 \item Compare the values obtained by each point $\vec{p}$ and keep
 the optimum one.
\end{enumerate}
Given the condition $W/\delta \in O(\text{poly}(n)), kM \in
O(\text{poly(n)})$, the first part of the algorithm can be done in
polynomial space. This is because in this case each point in the raw
set can be represented by polynomial space and therefore enumerated in polynomial space. The second part is more difficult. Computing
$\text{dis}(\vec{p}_t)$ for each t=1,...,k-1 can be done in
\class{NC(poly(n))} as shown in Lemma~\ref{lm:closeness_check}. Step
2 in the main algorithm only contains fundamental operations of
matrices and the spectrum decomposition. Thus, it also admits a
parallel algorithm in \class{NC(poly(n))}. One can easily compose
the two circuits and get a polynomial space implementation by
the relation \class{NC(poly)}=\class{PSPACE}~\cite{Borodin77}. The
third part can obviously be done in polynomial space.

Thus, by composing these three polynomial-space implementable parts,
one proves the whole algorithm can be done in \class{PSPACE}.
\end{proof}

\section{Simulation of several variants of QMA(2)} \label{sec:QMA_variants}
This section illustrates how one can make use of the algorithm shown
in Section~\ref{sec:main} (when $k$=2) to simulate some variants of
the complexity class \class{QMA(2)}. The idea is to show for those
variants, the corresponding POVM matrices of acceptance are
($M,\vec{w}$)-decomposable with small $M$s. Before we dive into the
details, let us recall the definition of the complexity class
\class{QMA(2)}.

\begin{definition} \label{def:QMA2}
A language $\mathcal{L}$ is in \class{QMA}$(2)_{m,c,s}$ if there
exists a polynomial-time generated family of quantum verification
circuits $Q=\{Q_n | n \in \mathbb{N}\}$ such that for any input $x$
of size $n$, the circuit $Q_n$ implements a two-outcome measurement
$\{Q^{\text{acc}}_x, \I-Q^{\text{acc}}_x\}$. Furthermore,
\begin{itemize}
  \item Completeness: If $x \in \mathcal{L}$, there exist 2 witness
  $\ket{\psi_1} \in \A_1, \ket{\psi_2} \in \A_2$, each of $m$ qubits, such that
  \[
    \ip{Q^{\text{acc}}_x}{\ketbra{\psi_1}{\psi_1} \ot
    \ketbra{\psi_2}{\psi_2}} \geq c.
  \]
  \item Soundness: If $x \notin \mathcal{L}$, then for any states $\ket{\psi_1} \in \A_1,
  \ket{\psi_2} \in \A_2$,
  \[
    \ip{Q^{\text{acc}}_x}{\ketbra{\psi_1}{\psi_1} \ot
    \ketbra{\psi_2}{\psi_2}} \leq s.
  \]
\end{itemize}
\end{definition}

We call \class{QMA(2)}=\class{QMA}$(2)_{\text{poly}(n),2/3,1/3}$. It
is easy to see that simulating the complexity class \class{QMA(2)}
amounts to distinguishing between the two promises of the maximum
acceptance probability, represented by the inner product
$\ip{Q^\text{acc}_x}{\rho}$,  over the set of all possible valid
strategies of the two provers, which is exactly
$\SepD{\A_1\ot\A_2}$. Note the maximum acceptance probability is
exactly $\OptSep{Q^\text{acc}_x}$ defined in
Problem~\ref{prob:opt_sep}. Thus, if one were able to distinguish
between the two promises of $\OptSep{Q^\text{acc}_x}$, one could
simulate this protocol with the same amount of resources (time or
space).

The first example is the variant with only logarithm-size proofs,
namely QMA$(2)_{O(\log(n)),2/3,1/3}$. It is not hard to find out the
corresponding POVMs of acceptance (i.e. $Q^\text{acc}_x$) need to be
(poly(n),$\vec{w}$)-decomposable since $\A_1,\A_2$ in this case are
only of polynomial dimension. Moreover, $\vec{w}$ could be $(1,1)$
in this case. Thus, it follows directly from
Corollary~\ref{cor:nc_main} that $\OptSep{Q^\text{acc}_x}$ can be
approximated in polynomial space. Namely,
\[
  \class{QMA}(2)_{O(\log(n)),2/3,1/3} \subseteq \class{PSPACE}.
\]

The next example is slightly less trivial. Before moving on, we need
some terminology about the quantum verification circuits $Q$. Assume
the input $x$ is fixed from now on. Let $\A_1,\A_2$ be the Hilbert
space of size $d_\A$ for the two proofs and let $\V$ be the
ancillary space of size $d_\V$. Note $d_\A d_\V$ is exponential in
$n$. Then the quantum verification process will be carried out on
the space $\A_1\ot \A_2 \ot \V$ with some initial state
$\ket{\psi_1} \ot \ket{\psi_2}\ot \ket{\vec{0}}$ where
$\ket{\psi_1}, \ket{\psi_2}$ are provided by the provers. The
verification process is also efficient in the sense that the whole
circuit only consists of polynomial elementary gates. Without loss
of generality, we can fix one universal gate set for the
verification circuits in advance. Particularly, we choose the
universal gate set to be single qubit gates plus the CNOT
gates~\cite{NielsenC00}. One can also choose other universal gate
sets without any change of the main result.

We categorize all elementary gates in the verification circuits into two
types. A gate is of \emph{type-I} if it only affects the qubits
within the same space (i.e, $\A_1,\A_2,\text{or},\V$). Otherwise,
this gate is of \emph{type-II}. It is easy to see single qubit gates
are always type-I gates. The only type-II gates are CNOT gates whose
control qubit and target qubit sit in different spaces. Let $p,r:
\mathbb{N} \rightarrow \mathbb{N}$ be polynomial-bounded functions.
A polynomial-time generated family of quantum verification circuits
$Q$ is called $Q[p,r]$ if each $Q_n$ only contains $p(n)$ type-I
elementary gates and $r(n)$ type-II elementary gates.

\begin{definition} \label{def:QMA2_gates}
 A language $\mathcal{L}$ is in
\class{QMA}$(2)_{m,c,s}[p,r]$ if $\mathcal{L}$ is in
\class{QMA}$(2)_{m,c,s}$ with some $Q[p,r]$ verification circuit
family.
\end{definition}

It is easy to see that $\class{QMA(2)}=\class{QMA(2)}[\text{poly,
poly}]$ from our definition. In the following we will show that when
the number of type-II gates is relatively small, one can simulate
this complexity model efficiently by the algorithm in
Fig.~\ref{fig:main_alg}.

\begin{lemma} \label{lm:POVM_polylog}
For any family of verification circuits $Q[p,r]$, the corresponding
POVM $Q^\text{acc}_x$ is $(4^{r(n)},(1,1))$-decomposable for any $n
\in \mathbb{N}$ and input $x$. Moreover, this decomposition can be
calculated in parallel with  $O(t(n) 4^{r(n)})\times \text{poly}(n)$
time.
\end{lemma}

\begin{proof} For any $n \in \mathbb{N}$ and input $x$, let us denote the whole
unitary that the verification circuit applies on the initial state by
$U=U_tU_{t-1}\cdots U_1$ where each $U_i$ corresponds to one elementary
gate and $t=p+r$. Without loss of generality, we assume the output
bit is the first qubit in the space $\V$ and the verification
accepts when that qubit is 1. Let $\bar{\V}$ be the space $\V$
without the first qubit, then we have
\[
  Q^\text{acc}_x = \tr_{\V} \( \I_{\A_1\A_2} \ot \ketbra{\vec{0}}{\vec{0}} (U^* \I_{\A_1\A_2} \ot \I_{\bar{\V}} \ot \ketbra{1}{1} U) \I_{\A_1\A_2} \ot
  \ketbra{\vec{0}}{\vec{0}}\).
\]
Let $P_{t+1}=\I_{\A_1\A_2} \ot \I_{\bar{\V}} \ot \ketbra{1}{1}$ and
$P_\tau=U^*_\tau P_{\tau+1} U_\tau$ for $\tau$=t,t-1,...,1. It is
easy to see $P_1=U^* (\I_{\A_1\A_2} \ot \I_{\bar{\V}} \ot
\ketbra{1}{1}) U$. Also it is straightforward to verify that
$P_{t+1}$ is 1-decomposable. Now let us observe how the
decomposability of $P_\tau$ changes with $\tau$.

For each $\tau$,  the unitary $U_\tau$ either corresponds to a type-I or
type-II elementary gate. In the former case, applying $U_\tau$ won't
change the decomposability. Thus, $P_\tau$ is $M$-decomposable
if $P_{\tau+1}$ is. In the latter case, applying $U_\tau$ will
potentially change the decomposability in the following sense. For
any such CNOT gate one has $U_\tau= \ketbra{0}{0}\ot \I +
\ketbra{1}{1} \ot X$ where $X$ is the Pauli matrix for the flip. And
one can show
\begin{eqnarray*}
  P_\tau  &= &(\ketbra{0}{0}\ot \I) P_{\tau+1} (\ketbra{0}{0}\ot \I)+(\ketbra{0}{0}\ot \I) P_{\tau+1} (\ketbra{1}{1}\ot
  X) \\
  & +& (\ketbra{1}{1}\ot X) P_{\tau+1} (\ketbra{0}{0}\ot \I)+(\ketbra{1}{1}\ot X) P_{\tau+1} (\ketbra{1}{1}\ot
  X).
\end{eqnarray*}
Thus in general we can only say $P_\tau$ is $4M$-decomposable
if $P_{\tau+1}$ is $M$-decomposable. As there are $r(n)$ type-II
gates, one immediately has $P_1$ is $4^{r(n)}$-decomposable.
Moreover, each operator appearing in the decomposition is a
multiplication of unitaries , $\ketbra{0}{0}, \ketbra{1}{1}$ and $X$
in some order, which implies the operator norm of those operators is
bounded by 1.  Therefore we have $P_1$ is ($4^{r(n)},
(1,1)$)-decomposable.

Finally, it is not hard to verify that multiplications with
$\I_{\A_1\A_2}\ot \ketbra{\vec{0}}{\vec{0}}$ and partial trace
over $\V$ won't change the decomposability of $P_1$. Namely, we have
$Q^\text{acc}_x$ is ($4^{r(n)},(1,1)$)-decomposable. The above proof
can also be considered as the process to compute the decomposition
of $Q^\text{acc}_x$. Each multiplication of matrices can be done in
\class{NC(poly(n))}. And the total number of multiplications is
upper bounded by $O(t(n) 4^{r(n)})$. Therefore, the total parallel
running time is upper bounded by $O(t(n) 4^{r(n)})\times$ poly($n$).
\end{proof}

\begin{corollary}
$\class{QMA(2)}[\text{poly}(n), O(\log(n))]$ $\subseteq
\class{PSPACE}$.
\end{corollary}

\begin{proof}
This is a simple consequence of
Lemma~\ref{lm:local_ham_decomposable} and
Corollary~\ref{cor:nc_main}. For any fixed $x$ of length $n$. One
can first compute the decomposition of $Q^\text{acc}_x$ in parallel
with $O(t(n) 4^{r(n)})\times$ poly($n$) time, which is parallel
polynomial time in $n$ when $r(n)=O(\log(n))$ and $t(n) \in
\text{poly}(n)$. Hence the first step can be done in polynomial
space via the relation
\class{NC(poly)}=\class{PSPACE}~\cite{Borodin77}.

Then one can invoke the parallel algorithm in
Corollary~\ref{cor:nc_main} to approximate $\OptSep{Q^\text{acc}_x}$
to sufficient precision $\delta$ such that one can distinguish
between the two promises. Precisely in this case, we choose those
parameters as follows,
\[
 k=2, W=1, M=4^{O(\log(n))}=\text{poly}(n), 1/\delta=\text{poly}(n).
\]
Thus the whole algorithm can be done in polynomial space, which
completes the proof.
\end{proof}

\noindent \textbf{Remarks}. Although the proof of the result is not
too technical, it establishes the first non-trivial upper bound
(\class{PSPACE} in this case) for variants of \class{QMA(2)} that
allow quantum operations acting on both proofs at the same time. In
contrast, previous results are all about variants with nonadaptive
or adaptive local measurements, like
\class{BellQMA(2)}~\cite{Bra08,ABD+09,CD10} or
\class{LOCCQMA(2)}~\cite{ABD+09,BCY11}.

However, our results are hard to extend to the most general case of
\class{QMA(2)}. This is because SWAP-test operation uses
many more type-II gates than what is allowed in our method. And SWAP-test
seems to be inevitable if one wants to fully characterize the power
of \class{QMA(2)}.

\section{Quasi-polynomial algorithms for local Hamiltonian cases} \label{sec:local_hamiltonian}
In this section, we illustrate that if $Q$ appears in the objective
function that is a local Hamiltonian then the optimal value $\OptSep{Q}$
can be efficiently computed by our main algorithm. Consider any
$k$-partite space $\A_1 \ot \A_2 \ot \cdots \ot \A_k$ where each
partite $\A_i$ contains $n$ qubits and thus is of dimension $2^n$.

\begin{definition} \label{def:local_hamiltonian}
Any Hermitian $Q$ over $\A_1 \ot \cdots \ot \A_k$ is a $l$-local
Hamiltonian if $Q$ is expressible as $Q=\sum_{i=1}^r H_i$ where each
term is a Hermitian operator acting on at most $l$ qubits among
$k$ parties.
\end{definition}

Hamiltonians are widely studied in physics since they usually
characterize the energy of a physical system. Local Hamiltonians
are of particular interest since they refer to the energy of many
interesting models in low-dimension systems. Our algorithm can be
considered as a way to find the minimum energy in the system
achieved by separable states.

Local Hamiltonians are also appealing to computational complexity
theorists since the discovery of the promise 5-local Hamiltonian
problem~\cite{KitaevW02} which turns out to be \class{QMA}-complete.
Precisely, it refers to the following promise problem when $k=1,l=5$.

\begin{problem}[$k$-partite $l$-local Hamiltonian problem] \label{prob:promise_local_Ham}
Take the expression $Q=\sum_{i=1}^r H_i$ for any $l$-local
Hamiltonian over $\A_1 \ot \cdots \ot \A_k$ as input\footnote{It is
noteworthy to mention that the input size of local Hamiltonian
problems can be only poly-logarithm in the dimension of the space
where $Q$ sits in. }, where
$\opnorm{H_i}\leq 1$ for each $i$. Let $\OptSep{Q}$ denote the
minimum value of $\ip{Q}{\rho}$ achieved for some $\rho \in
\SepD{\A_1 \ot \cdots \ot \A_k}$. The goal is to tell between the
following two promises: either $\OptSep{Q} \geq a$ or $\OptSep{Q}
\leq b$ for some $a>b$ with inverse polynomial gap.
\end{problem}

When $k=1$, the promise problem defined above is exactly the
original $l$-local Hamiltonian problem. Subsequent results
demonstrate that it remains \class{QMA}-complete even
when $l=3,2$~\cite{AharonovGIK09,KempeKR06,OliveriaT08}. Our
definition of the promise problem naturally extends to the
$k$-partite case. We refer to Chapter 14 in~\cite{KitaevW02}
for technical details. It is not hard to see that $k$-partite
$l$-local Hamiltonian problems belong to \class{QMA(k)} by applying
similar techniques in the original proof. However, they do not remain as \class{QMA(k)}-complete problems.
This is because the original reduction transforms from the proof space to the transcript and clock
space and the separability of quantum states does not persevere
under such an operation. As a result,  $k$-partite local Hamiltonian
problems defined above only enforce the separability in the
transcript and clock space rather than in the proof space. Note this
is not an issue for the 1-partite case since there is no
separability involved. Nevertheless it becomes a huge problem for
its $k$-partite extensions.

\begin{lemma} \label{lm:local_ham_decomposable}
Any $l$-local Hamiltonian $Q$ over $\A_1 \ot \cdots \ot \A_k$ such
that $Q=\sum_{i=1}^r H_i$ and $\opnorm{H_i} \leq w$ is
($O((4nk)^l),w$ )-decomposable.
\end{lemma}

\begin{proof}
Since $Q$ is a $l$-local Hamiltonian, it is easy to see $r \leq
{ kn \choose l}$. For each $H_i$ with $\opnorm{H_i} \leq w$, since
it acts only on at most $l$ qubits, it must be
($4^l,w$)-decomposable. Thus $Q$ is ($r4^l, w$)-decomposable.
In terms of only $n,k,l$, we have $Q$ is ($O((4nk)^l),w$)-decomposable.
\end{proof}

\begin{corollary} \label{cor:local_ham}
Take the expression $Q=\sum_{i=1}^r H_i$ of any $l$-local
Hamiltonian  over $\A_1 \ot \cdots \ot \A_k$ (each $\A_i$ is of
dimension $d=2^n$) such that $\opnorm{H_i} \leq w$ for each $i$ as input.
Assuming $k, l =O(1)$, the quantity $\OptSep{Q}$ can be approximated to
precision $\delta$ in quasi-polynomial time in $d, w, 1/\delta$.

If $n$ is considered as the input size and
$w/\delta=O(\text{poly(n)})$, then $\OptSep{Q}$ can be approximated
to precision $\delta$ in \class{PSPACE}.
\end{corollary}

\begin{proof}
The proof of the first part follows directly from
Lemma~\ref{lm:local_ham_decomposable} and
Theorem~\ref{thm:main_alg}. Recall the proof of
Lemma~\ref{lm:local_ham_decomposable} also provides a way to compute
the decomposition of $Q$ given the expression $Q=\sum_{i=1}^r H_i$
as input. It is easy to verify that $O(r4^l)$ time (upper bounded by
$O((4k\log d)^l)$) is sufficient to complete this computation.
After that, one may directly invoke the algorithm in
Fig.~\ref{fig:main_alg} and make use of Theorem~\ref{thm:main_alg}.
Now we substitute the following identities into our main algorithm. Note
$k,l=O(1)$ and we have
$M=O(\log^{O(1)} d), W=w^{O(1)}$.
One immediately gets the total running time bounded by
\[
 \exp( O( \log^{O(1)} (d) ( \log\log d + \log w/\delta )))
 \times \text{poly}(d, w, 1/\delta),
\]
which is quasi-polynomial time in $d, w, 1/\delta$.

For the second part when $n$ is considered as the input size, it is
easy to see the computation of the decomposition of $Q$
according to Lemma~\ref{lm:local_ham_decomposable} can be done in
\class{NC(poly)}, henceforth in polynomial space. (Note $M=O(\text{poly(n)})$.) Then by composing with the polynomial-space algorithm implied by
Corollary~\ref{cor:nc_main}, one proves the whole algorithm can be
implemented in polynomial space.
\end{proof}

\noindent \textbf{Remarks.} It is a direct consequence of
Corollary~\ref{cor:local_ham} that
Problem~\ref{prob:promise_local_Ham} is inside \class{PSPACE}.

\section{An algorithm with running time exponential in $\fbnorm{Q}$ } \label{sec:alg_Q2}
In this section  we demonstrate another application of the
simple idea "enumeration" by epsilon-net to
Problem~\ref{prob:opt_sep}. As a result, we obtained an algorithm
with running time exponential in $\fbnorm{Q}$ (or
$\loccnorm{Q}~\cite{MWW09}$\footnote{This follows easily from the
fact $\fbnorm{Q}=O(\loccnorm{Q})$~\cite{MWW09} where $\loccnorm{Q}$ stands for the LOCC norm of the operator Q. })
for computing $\OptSep{Q}$ with additive error $\delta$. A similar
running time $\exp(O(\log^2(d)\delta^{-2}\fbnorm{Q}^2 ))$ was
obtained in~\cite{BCY11} using some known results in quantum
information theory.(i.e., the semidefinite programming for finding
symmetric extension~\cite{DPS04} and an improved quantum de
Finetti-type bound.)

By contrast, our algorithm makes no use of any advanced tool above
and only utilizes fundamental operations of matrices. Intuitively, in order
to approximate the optimum value to precision $\delta$, one only
needs to look at the eigenspace of eigenvalues
greater than $\delta$, the dimension of which is no more than
$\fbnorm{Q}^2/\delta^2$.
Nevertheless, naively enumerating density operators over that
subspace doesn't work since one cannot detect the separability of
those density operators. We circumvent this difficulty by making
nontrivial use of the Schmidt decomposition of bipartite pure states.

Finally, as mentioned in the introduction we admit that other results in the original
paper~\cite{BCY11} do not follow from our algorithm and our method
cannot be seen as a replacement of the kernel technique of that
paper. Also our method does not extend to the $k$-partite version as there is no Schmidt decomposition in that case.

\begin{figure}[t]
\noindent\hrulefill
\begin{mylist}{8mm}
\item[1.]
Compute the spectral decomposition of $Q$. After that, one has the
decomposition $Q=\sum_{t} \lambda_t \ketbra{\Psi_t}{\Psi_t}$. Choose
$\epsilon=\delta/2$ and $\Gamma_\epsilon=\{t: \lambda_t \geq \epsilon\}$.
 Also let $\text{OPT}$ store the optimum value of the maximization
problem.
\item[2.]
Generate the $\varepsilon$-net of the unit ball of
$\complex^{|\Gamma_\epsilon|}$ under the Euclidean norm with
$\varepsilon= \frac{\delta}{4\fbnorm{Q}}$. Denote such set by
$\N_\varepsilon$. Then for each point $\alpha \in \N_\varepsilon$,
\begin{mylist}{8mm}
\item[(a)] Compute $\ket{\phi_\alpha}=\sum_{ t \in \Gamma_\epsilon} \alpha^*_t
\sqrt{\lambda_t} \ket{\Psi_t}$ and compute the Schmidt decomposition
of $\ket{\phi_\alpha}$, i.e.
\[
  \ket{\psi_\alpha} =\sum_i \mu_i \ket{u_i} \ket{v_i},
\]
where $\mu_1 \geq \mu_2 \geq \cdots$ and $\{u_i\}, \{v_i\}$ are
orthogonal bases. Note $\ket{\phi_\alpha}$ is not necessarily a unit
vector.
\item[(b)] Update OPT as follows: OPT=$\max\{$OPT,$\mu_1\}$.
\end{mylist}
\item[3.]
Return $\text{OPT}$.
\end{mylist}
\noindent\hrulefill \caption{The algorithm runs in time exponential
in $\fbnorm{Q}/\delta$. } \label{fig:LOCC_alg}
\end{figure}

Recall the optimization problem we are interested in is equivalent
to the following one.
\[
  \max: \ip{Q}{\rho} \text{ s.t. }
  \rho=\ketbra{u}{u}\ot\ketbra{v}{v}, \ket{u} \in \A_1, \ket{v} \in
  \A_2.
\]

\begin{theorem} \label{thm:LOCC_alg}
Given any positive semidefinite $Q$ over $\A_1 \ot \A_2$ (of
dimension $d \times d$) and $\delta>0$, the algorithm in
Fig.~\ref{fig:LOCC_alg} approximates the optimal value $\OptSep{Q}$
with additive error $\delta$ with running time
$\exp(O(\log(d)+\delta^{-2} \fbnorm{Q}^2 \ln (\fbnorm{Q}/\delta)))$.
\end{theorem}

\begin{proof}
We first prove the correctness of the algorithm. The analysis will
mainly be divided into two parts. Let $S_\epsilon=\text{span}\{
\ket{\Psi_t} | t \in \Gamma_\epsilon\}$. The first part shows it
suffices to only consider vectors inside the subspace $S_\epsilon$
for approximating $\OptSep{Q}$ with additive error $\delta$. The
second one demonstrates that our algorithm in
Fig.~\ref{fig:LOCC_alg} approximates the optimal value obtained by
only considering vectors in $S_\epsilon$. Precisely, since $\{
\ket{\Psi_i}\}$ forms a basis, one has  $\ket{u}\ket{v}=\sum_{t \in
\Gamma_\epsilon} \beta_t \ket{\Psi_t} + \sum_{t \notin
\Gamma_\epsilon} \beta_t \ket{\Psi_t}$ where $\beta$ is a unit
vector in $\complex^{d^2}$. Then we have
\[
  \ip{Q}{\ketbra{u}{u}\ot \ketbra{v}{v}}= \underbrace{\sum_{t \in \Gamma_\epsilon}
  \lambda_t |\beta_t|^2}_{(I)} + \underbrace{\sum_{t \notin \Gamma_\epsilon} \lambda_t
  |\beta_t|^2}_{(II)},
\]
where the term (II) is obviously bounded by $\delta/2$ (i.e.,
$\sum_{t \notin \Gamma_\epsilon} \lambda_t |\beta_t|^2 \leq
\delta/2$). For the term (I), it is equivalent to $\OptSep{\tilde{Q}}$ where $\tilde{Q}=\sum_{t \in
\Gamma_\epsilon}
  \lambda_t \ketbra{\Psi_t}{\Psi_t}$. Namely, small eigenvalues are
truncated in $\tilde{Q}$. Now observe the following identity.
\begin{eqnarray*}
  \max_{\ket{u}\ket{v}} \ip{\tilde{Q}}{\ketbra{u}{u} \ot
  \ketbra{v}{v}} & = & \max_{\ket{u}\ket{v}} \sum_{t \in \Gamma_\epsilon} \lambda_t
  |\bra{u}\braket{v}{\Psi_t}|^2 = \max_{\ket{u}\ket{v}} \snorm{\gamma^{u,v}}^2 \\
  & = & \max_{\ket{u}\ket{v}} \max_{\alpha \in
  \ball{|\Gamma_\epsilon|}{}} |\sum_{t \in \Gamma_\epsilon} \alpha^*_t
  \sqrt{\lambda_t}\bra{u}\braket{v}{\Psi_t}|^2 =\max_{\ket{u}\ket{v}} \max_{\alpha \in
  \ball{|\Gamma_\epsilon|}{}} |\bra{u}\braket{v}{\phi_\alpha}|^2 \\
  & = & \max_{\alpha \in
  \ball{|\Gamma_\epsilon|}{}}\max_{\ket{u}\ket{v}}
  |\bra{u}\braket{v}{\phi_\alpha}|^2,
\end{eqnarray*}
where $\gamma^{u,v} \in \complex^{|\Gamma_\epsilon|}$ and
$\gamma^{u,v}_t= \sqrt{\lambda_t} \bra{u}\braket{v}{\Psi_t}$ for
each $t \in \Gamma_\epsilon$. The second line comes from the duality
of the Euclidean norm (i.e., $\snorm{y}=\max_{ \snorm{z} \leq 1}
|\braket{z}{y}|$). The third line comes by exchanging positions of
the two maximizations. We then make use of the following well-known
fact. 

\begin{fact}[\cite{NielsenC00}]
For any bipartite vector $\ket{\psi}$ with the Schmidt decomposition
\[
  \ket{\psi} =\sum_i \mu_i \ket{u_i} \ket{v_i},
\]
where $\mu_1 \geq \mu_2 \geq \cdots$ and $\{u_i\}, \{v_i\}$ are
orthogonal bases . Then $\max_{\ket{u}\ket{v}}
|\bra{u}\braket{v}{\psi}|=\mu_1$ and the maximum value is obtained
by choosing $\ket{u}\ket{v}$ to be $\ket{u_1}\ket{v_1}$.
\end{fact}

It is not hard to see that our algorithm computes exactly the term
on the third line except that we replace the unit ball by its
$\epsilon$-net. However, this won't incur too much extra error. For any
$\alpha \in \ball{|\Gamma_\epsilon|}{}$, there exists
$\tilde{\alpha} \in \N_\epsilon$, such that
$\snorm{\alpha-\tilde{\alpha}} \leq \varepsilon$. Thus, the extra
error incurred is $||\bra{u}\braket{v}{\phi_\alpha}|^2-
|\bra{u}\braket{v}{\phi_{\tilde{\alpha}}}|^2|$ and can be bounded by
\begin{eqnarray*}
  (\snorm{\ket{\phi_\alpha}}+\snorm{\ket{\phi_{\tilde{\alpha}}}})|\bra{u}\braket{v}{\psi_\alpha-\psi_{\tilde{\alpha}}}|
  & \leq & 2 \max_{\snorm{\beta_1}\leq 1}
  \snorm{\phi_{\beta_1}} \max_{\beta_2=\alpha-\tilde{\alpha}, \snorm{\beta_2}\leq \varepsilon} \snorm{\phi_{\beta_2}} \\
  & = & 2\sqrt{\fbnorm{Q}}  \times \varepsilon  \sqrt{\fbnorm{Q}} \leq
  \delta/2,
\end{eqnarray*}
where $\max_{\snorm{\beta}\leq \epsilon'}
  \snorm{\phi_{\beta}} \leq  \epsilon' \sqrt{\fbnorm{Q}}$ for any $\epsilon'>0$ can be verified directly and therefore the total additive error is bounded by
$\delta/2+\delta/2=\delta$.

Finally, let us turn to the analysis of the efficiency of this
algorithm. The spectrum decomposition in the first step takes
polynomial time in $d$, so is the same with calculation of
$\ket{\psi_\alpha}$. The generation of the $\epsilon$-net of the unit ball
is standard and can be done in
$O((1+\frac{2}{\varepsilon})^{|\Gamma_\epsilon|})\times
\text{poly}(|\Gamma_\epsilon|)$. The last operation, finding the Schmidt
decomposition, is equivalent to singular value decompositions, and
thus can be done in polynomial time in $d$ as well. Also note
$|\Gamma_\epsilon| \leq \min \{d^2, \fbnorm{Q}^2/\delta^2\}$. To sum
up, the total running time of the algorithm is upper bounded by
$O((1+\frac{2}{\varepsilon})^{|\Gamma_\epsilon|})\times
\text{poly}(d)$, or equivalently $\exp(O(\log(d)+\delta^{-2}
\fbnorm{Q}^2 \ln (\fbnorm{Q}/\delta)))$.

\end{proof}

\noindent \textbf{Remarks.} One can also apply the observation in
the introduction to parallelize the computation in this case.
However, the size of the $\epsilon$-net here will depend on some
parameter (i.e. $\fbnorm{Q}/\delta$) other than the input.

\subsection*{Acknowledgement}
We thank Zhengfeng Ji and John Watrous for helpful discussions. This
research was supported in part by National Basic Research
Program of China Awards 2011CBA00300 and 2011CBA00301, and by
NSF of United States Award 1017335.

\bibliographystyle{alpha}

\appendix

\section{Proof of Lemma~\ref{lm:closeness_check}} \label{app:proof_lm_close}

\begin{theorem}[Multiplicative weights update method---see Ref.~{\cite[Theorem 10]{Kale07}}]
\label{thm:mwum}

Fix $\gamma\in(0,1/2)$. Let $N^{(1)},\dots,N^{(T)}$ be arbitrary
$d\times d$ ``loss'' matrices with $0\preceq N^{(t)}\preceq \alpha
I$. Let $W^{(1)},\dots,W^{(T)}$ be $d\times d$ ``weight'' matrices
given by
\begin{align*}
  W^{(1)} &= I &
  W^{(t+1)} &= \exp(-\gamma(N^{(1)} + \cdots + N^{(t)}) ).
\end{align*}
Let $\rho^{(1)},\dots,\rho^{(T)}$ be density operators obtained by
normalizing each $W^{(1)},\dots,W^{(T)}$ so that
$\rho^{(t)}=W^{(t)}/\tr{W^{(t)}}$. For all density operators $\rho$
it holds that
\[
  \frac{1}{T}\sum_{t=1}^T \ip{\rho^{(t)}}{N^{(t)}} \leq
  \ip{\rho}{\frac{1}{T}\sum_{t=1}^T N^{(t)}} + \alpha(\gamma + \frac{\ln d}{\gamma T}).
\]
\end{theorem}

Note that Theorem \ref{thm:mwum} holds for \emph{all} choices of
loss matrices $N^{(1)},\dots,N^{(T)}$, including those for which
each $N^{(t)}$ is chosen adversarially based upon
$W^{(1)},\dots,W^{(t)}$. This adaptive selection of loss matrices is
typical in implementations of the MMW. Consider the algorithm shown
in Fig.~\ref{fig:closeness}.

\begin{figure}[t]
\noindent\hrulefill
\begin{mylist}{8mm}
\item[1.]
Let $\gamma=\frac{\epsilon}{8Mw}$ and $T=\ceil{\frac{\ln
d}{\gamma^2}}$. Also let $W^{(1)}=\I_{\X}$, $d=\dim{(\X)}$.
\item[2.]
Repeat for each $t = 1,\ldots,T$:

\begin{mylist}{8mm}
\item[(a)]
Let $\rho^{(t)}=W^{(t)}/\tr{W^{(t)}}$ and compute
$\vec{q}(\rho^{(t)})$ (\citeequ{\ref{eqn:def_vec_q}}). One can then
rewrite the vector $\vec{p}-\vec{q}(\rho^{(t)})$ in the polar form
$(c^{(t)}_1e^{i\phi^{(t)}_1}, c^{(t)}_2e^{i\phi^{(t)}_2},\cdots,
c^{(t)}_M e^{i\phi^{(t)}_M})$ and choose
$\vec{z}^{(t)}=(e^{-i\phi^{(t)}_1}, e^{-i\phi^{(t)}_2}, \cdots,
e^{-i\phi^{(t)}_M})$. It is not hard to see such $\vec{z}^{(t)}$
maximizes $\Rep \ip{\vec{p}-\vec{q}(\rho^{(t)})}{\vec{z}}$.

\item[(b)]
Choose $N^{(t)}$ to be
\[
   N^{(t)}= \Rep \ip{\vec{p}}{\vec{z}^{(t)}} \I_\X-
   \frac{1}{2}(Q^{(t)}+Q^{(t)*})+2Mw\I_\X,
\]
where $Q^{(t)}=\sum_{i=1}^M e^{+i\phi^{(t)}_i}Q_i$.
\item[(c)]
Update the weight matrix as follows: $ W^{(t+1)}=\exp(-\gamma
\sum_{\tau=1}^t N^{(\tau)})$.
\end{mylist}

\item[3.]
Return $\tilde{d}=\frac{1}{T} \sum_{t=1}^T
\ip{\rho^{(t)}}{N^{(t)}-2Mw\I_\X}$.
\end{mylist}
\noindent\hrulefill \caption{An algorithm that approximates the
$d(\vec{p})$ with additive error $\epsilon$. } \label{fig:closeness}
\end{figure}

\begin{lemma}[Restated Lemma~\ref{lm:closeness_check}]
Given any point $\vec{p} \in \sraw{M,w}$ and $\epsilon>0$, the
algorithm in Fig.~\ref{fig:closeness} approximates
$\text{dis}(\vec{p})$ with additive error $\epsilon$. Namely, the
return value $\tilde{d}$ of this algorithm satisfies
\[
   \tilde{d}-\epsilon \leq \text{dis}(\vec{p}) \leq
   \tilde{d}+\epsilon.
\]
Moreover, the algorithm runs in $\text{poly}(d,M,w,1/\epsilon)$
time. Furthermore, if $d$ is considered as the input size and
$M,w,1/\epsilon \in O( \text{poly-log}(d))$, this algorithm is also
efficient in parallel, namely, inside \class{NC}.
\end{lemma}

\begin{proof}
The algorithm is a typical application of the matrix multiplicative
weight update method. In order to make use of
Theorem~\ref{thm:mwum}, we need first to show $N^{(t)}$ is bounded
for each $t$.  Since $\vec{p} \in \sraw{M,w}$ and
$\snorm{\vec{z}^{(t)}}_\infty \leq 1$, by Cauchy-Schwartz inequality
we have
\[
  |\Rep \ip{\vec{p}}{\vec{z}^{(t)}}| \leq \snorm{\vec{p}}_1
  \snorm{\vec{z}^{(t)}}_\infty \leq M \snorm{\vec{p}}_\infty
  \snorm{\vec{z}^{(t)}}_\infty= Mw.
\]
Furthermore we have
\[
  \snorm{Q}_\text{op} = \snorm{\sum_{i=1}^M e^{-i\phi^{(t)}_i}Q_i
  }_\text{op} \leq \sum_{i=1}^M \snorm{Q_i}_\text{op} \leq Mw.
\]
Thus by triangle inequality, one can easily find
\[
   0\preceq N^{(t)}\preceq 4Mw\I_\X.
\]
Then we can make use of Theorem~\ref{thm:mwum}. Immediately, for
any $\rho \in \density{\X}$, we have
\[
  \frac{1}{T}\sum_{t=1}^T \ip{\rho^{(t)}}{N^{(t)}} \leq
  \ip{\rho}{\frac{1}{T}\sum_{t=1}^T N^{(t)}} + \alpha(\gamma + \frac{\ln d}{\gamma T}).
\]
Substitute $\alpha=4Mw, \gamma=\frac{\epsilon}{8Mw} $ and
$T=\ceil{\frac{\ln d}{\gamma^2}} $. Also consider the identity
$\ip{\rho^{(t)}}{N^{(t)}-2Mw\I_\X}= \Rep
\ip{\vec{p}-\vec{q}(\rho^{(t)})}{\vec{z}^{(t)}}$. Then we have for
any $\rho \in \density{\X}$,
\begin{equation} \label{eqn:mmw_proof}
\tilde{d}=\frac{1}{T} \sum_{t=1}^T \ip{\rho^{(t)}}{N^{(t)}-2Mw\I_\X}
\leq \Rep
\ip{\vec{p}-\vec{q}(\rho)}{\frac{1}{T}\sum_{t=1}^T\vec{z}^{(t)}} +
\epsilon.
\end{equation}
Consider the equilibrium value form of $\text{dis}(\vec{p})$ in
\citeequ{\ref{eqn:closeness_equilibrium}}. For each $\rho^{(t)}$, we
always find the $\vec{z}^{(t)}$ that maximizes $\Rep
\ip{\vec{p}-\vec{q}(\rho^{(t)})}{\vec{z}}$. Hence,
$\text{dis}(\vec{p})\leq \tilde{d}$. Let $\rho^\star$ be any
equilibrium point of the equilibrium value in
\citeequ{\ref{eqn:closeness_equilibrium}}. By substituting such
$\rho^\star$ into \citeequ{\ref{eqn:mmw_proof}} we have
\[
 \tilde{d}\leq \Rep
\ip{\vec{p}-\vec{q}(\rho^\star)}{\frac{1}{T}\sum_{t=1}^T\vec{z}^{(t)}}
+ \epsilon \leq \text{dis}(\vec{p})+\epsilon.
\]
So far we complete the proof of the correctness of this algorithm.
Note that each step in the algorithm only contains fundamental
operations of matrices and vectors, which can be done in polynomial
time in $M,d$. Also there are totally $O(T)=\text{poly}(\ln
d,M,w,1/\epsilon)$ steps, thus the whole algorithm can be executed
in $\text{poly}(d,M,w,1/\epsilon)$ time.  Moreover, given the fact
that fundamental operations of matrices and vectors also admit
efficient algorithms in parallel (i.e., \class{NC} algorithm), one
can easily compose these \class{NC} circuits of each step and
obtain a \class{NC} algorithm as a whole if the total number of
steps $T$ is not too large. Precisely, if $M,w,1/\epsilon \in O(
\text{poly-log}(d))$ and  $d$ is considered as the input size, this
algorithm is also efficient in parallel.
\end{proof}

\end{document}